\documentclass[11pt]{amsart}

\usepackage{amsmath, amssymb, amsthm}
\usepackage{hyperref}
\usepackage[backend=biber,sorting=none]{biblatex}
\usepackage{listings}
\usepackage{xcolor}

\theoremstyle{plain}
\newtheorem{theorem}{Theorem}[section]
\newtheorem{lemma}[theorem]{Lemma}
\newtheorem{conjecture}[theorem]{Conjecture}
\newtheorem{proposition}[theorem]{Proposition}

\theoremstyle{definition}
\newtheorem{definition}[theorem]{Definition}

\theoremstyle{remark}
\newtheorem{remark}[theorem]{Remark}

\title[Spectral Entropy of Distance Distributions]{A Scale-Invariant Entropy Statistic for Distance Distributions}
\author{Mohamed Gewily}

\begin{document}

\begin{abstract}
We introduce a family of scale-invariant entropy statistics derived from logarithmically aggregated distance distributions of point processes, with prime numbers serving as a motivating example.
The construction associates to each finite configuration a scalar quantity encoding structural features of relative spacing while remaining insensitive to absolute scale.
This work is intended as a methodological contribution rather than a source of new raw results.
\end{abstract}

\maketitle

\section{Overview and Objectives}

This paper introduces a scale-invariant entropy observable designed to quantify structural features of distance distributions arising from discrete point configurations on the real line. The central idea is to represent inter-point distances on a logarithmic axis at fixed resolution.
This eliminates dependence on absolute scale while preserving relative spacing information.
Discrete harmonic analysis on the resulting log-distance distribution produces a finite spectrum whose frequency content reflects regularity or randomness.
Compressing this spectrum via entropy yields a robust scalar statistic suitable for comparison across scales, configurations, and null models. 

The construction proceeds through four steps:
\begin{enumerate}
  \item extraction of truncated distance data,
  \item logarithmic aggregation at fixed resolution,
  \item discrete harmonic analysis on the log-distance coordinate,
  \item entropy-based compression of spectral information.
\end{enumerate}

The framework is conducted at fixed logarithmic binning rather than under asymptotic refinement.
Fixing the binning defines the observational scale, yielding numerically stable, interpretable quantities that can be directly compared with stochastic reference models. Using prime numbers as a motivating example, we study the expected behavior of the statistic, and we conjecture about its convergence.
\section{Prime Distance Distributions}

Let $\mathcal P = \{p_1 < p_2 < \dots\}$ denote the set of prime numbers.

\begin{definition}[Truncated Prime Distance Measure]
Fix a prime $p$ and a truncation radius $R>0$.
Define the truncated distance multiset
\[
D_R(p) := \{\, |p-q| : q \in \mathcal P,\ q \neq p,\ |p-q| \le R \,\}.
\]
We identify this multiset with its associated counting measure on $\mathbb R_{>0}$.
\end{definition}

\begin{definition}[Additivity]
For a finite multiset of primes $\mathbf p = \{p_1,\dots,p_k\}$, define
\[
D_R(\mathbf p) := \sum_{j=1}^k D_R(p_j),
\]
where addition denotes addition of measures.
\end{definition}

\begin{remark}
Additivity enables aggregation of spacing information across multiple base points, permitting the study of structured collections such as prime tuples or other composite configurations.
\end{remark}

\section{Logarithmic Aggregation at Fixed Resolution}

Absolute prime distances vary strongly with scale and are not directly comparable across magnitude regimes.
To address this, distances are represented on a logarithmic axis with fixed resolution.

Fix an integer $M \ge 2$.
Let
\[
d_{\min} := \min D_R(\mathbf p), \qquad d_{\max} := \max D_R(\mathbf p),
\]
assuming $D_R(\mathbf p)$ is nonempty.

Define logarithmic bin edges
\[
\ell_j := \log d_{\min} + \frac{j}{M}(\log d_{\max}-\log d_{\min}),
\quad j=0,\dots,M,
\]
and set $b_j := e^{\ell_j}$.

Define bins
\[
B_j := [b_{j-1},b_j), \ j=1,\dots,M-1, \quad B_M := [b_{M-1},b_M].
\]

Let
\[
n_j := \#\{ d \in D_R(\mathbf p) : d \in B_j \}, \qquad
p_j := \frac{n_j}{\sum_{k=1}^M n_k}.
\]

\begin{remark}[Scale Invariance]
Multiplying all distances by a positive constant leaves the probability vector $(p_1,\dots,p_M)$ unchanged, up to discretization effects at bin boundaries.
\end{remark}

\begin{remark}[Resolution Dependence]
Fixing $M$ emphasizes structure visible at a prescribed logarithmic resolution and suppresses fine-scale fluctuations while preserving scale invariance.
\end{remark}

\section{Log-Frequency Spectrum}

Define logarithmic bin centers
\[
x_j := \frac{\ell_{j-1}+\ell_j}{2}.
\]

Define the discrete log-frequency spectrum
\[
\widehat{\mu}(k)
=
\sum_{j=1}^M
p_j
\exp\!\Big(
-2\pi i (k-1)
\frac{x_j-x_1}{x_M-x_1}
\Big),
\quad k=1,\dots,M.
\]

\begin{remark}
Using logarithmic bin centers rather than indices ensures that the spectrum reflects the geometry of the log-distance axis rather than discretization artifacts.
Spectral concentration at low frequencies corresponds to slowly varying structure, while broad spectral support indicates irregular or random behavior.
\end{remark}

\section{Spectral Entropy}

\begin{definition}[Spectral Entropy]
Define normalized spectral weights
\[
w_k := \frac{|\widehat{\mu}(k)|}{\sum_{\ell=1}^M |\widehat{\mu}(\ell)|},
\quad k=1,\dots,M.
\]
The spectral entropy is
\[
H := -\sum_{k:w_k>0} w_k \log w_k.
\]
\end{definition}

\begin{remark}
Linear magnitudes moderates sensitivity to a single dominant frequency.
Alternative choices, such as squared magnitudes, yield similar qualitative behavior.
\end{remark}

\section{Null Models and Interpretation}

Regular configurations produce spectra dominated by low frequencies and correspondingly low entropy.
Random point processes yield flatter spectra and higher entropy.
These regimes provide natural reference points for interpreting empirical entropy values.

\section{Poisson Log-Distance Model}

Let $\Pi$ be a homogeneous Poisson point process on $\mathbb R_{>0}$ with intensity $\lambda>0$, conditioned to contain $0$.
Define
\[
D_R^{\mathrm{null}} := \{x \in \Pi : 0<x\le R\}.
\]
\begin{lemma}[Asymptotic Stabilization of Logarithmic Bins]
Under the Poisson log-distance model, almost surely
\[
\log d_{\max} = \log R + o(1),
\qquad
\log d_{\min} = O(1),
\quad \text{as } R \to \infty.
\]
Consequently,
\[
\frac{\log d_{\max}-\log d_{\min}}{M}
=
\frac{\log R}{M} + O(1),
\]
and the logarithmic bin geometry stabilizes asymptotically.
\end{lemma}

\begin{proof}[Proof Sketch]
Since the process is conditioned to contain $0$, the smallest positive point
has exponential distribution with parameter $\lambda$ and is therefore
$O_{\mathrm{a.s.}}(1)$ as $R \to \infty$, yielding $\log d_{\min}=O(1)$ almost surely.
The largest point of a homogeneous Poisson process in $(0,R]$ lies within
$O_{\mathrm{a.s.}}(1)$ of $R$ by standard renewal properties of Poisson processes,
so $d_{\max}=R-\Delta_R$ with $\Delta_R=O_{\mathrm{a.s.}}(1)$.
Hence
\[
\log d_{\max}
=
\log(R-\Delta_R)
=
\log R + o(1).
\]
The stated asymptotics follow immediately.
\end{proof}

\begin{proposition}[Entropy Convergence for Poisson Log-Distance Model]
Fix $M \ge 2$.
There exists a constant $H_\infty^{\mathrm{null}}(M)$ such that
\[
H_R^{\mathrm{null}} \xrightarrow{\mathrm{a.s.}} H_\infty^{\mathrm{null}}(M)
\quad \text{as } R \to \infty.
\]
Convergence also holds in $L^1$, and the limit is independent of $\lambda$.
\end{proposition}

\begin{proof}[Proof Sketch]
Under the transformation $x \mapsto \log x$, the Poisson process becomes inhomogeneous with intensity $\lambda e^t\,dt$.
Although the bin edges depend on the random extrema $d_{\min}$ and $d_{\max}$,
the preceding lemma shows that the logarithmic partition has asymptotically stable geometry.
Hence the relative logarithmic bin structure becomes asymptotically deterministic at scale
$\log R$, and the strong law of large numbers applies to each bin.
Continuity of the discrete Fourier transform and entropy maps on the probability
simplex then implies convergence of $H_R^{\mathrm{null}}$.
\end{proof}

\section{Connection to Classical Prime Gap Models}

Cramér’s random model \cite{cramer_order_1936} predicts Poissonian behavior for suitably rescaled prime gaps.
Under Hardy--Littlewood-type assumptions \cite{hardy_littlewood_1923}, Gallagher \cite{gallagher_distribution_1976} established rigorous Poisson limits for gap statistics.

\begin{proposition}[Entropy Universality under Cramér's Model]
Fix $M\ge 2$.
Under Cramér’s model, after rescaling by the local mean gap, the spectral entropy converges in probability to
\[
H_\infty^{\mathrm{null}}(M).
\]
\end{proposition}

\begin{proof}[Proof Sketch]
Rescaled prime gaps converge in distribution to a homogeneous Poisson process.
Logarithmic aggregation removes scale dependence, and continuity of the entropy functional yields convergence to the Poisson null-model value.
\end{proof}

\section{Questions and Conjectures}

The following conjectures describe stability and limiting behavior of the
spectral entropy observable at fixed logarithmic resolution.

\begin{conjecture}[Asymptotic Entropy Stability]
Fix a prime $p$ and resolution $M$.
There exists a function $\eta(R)\to 0$ such that
\[
|H_{R_1}(p)-H_{R_2}(p)| \le \eta(\min(R_1,R_2))
\]
for all sufficiently large $R_1,R_2$.
\end{conjecture}

\begin{conjecture}[Entropy Deviation from Poisson Statistics]
For fixed $M$, the limit
\[
\Delta_\infty(p)
:=
\lim_{R\to\infty}
\big(H_R(p)-H_R^{\mathrm{null}}\big)
\]
exists.
If prime gaps exhibit persistent deviations from Poisson behavior
on the logarithmic scale, then $\Delta_\infty(p)\neq 0$
for infinitely many primes.
\end{conjecture}

\begin{conjecture}[Ensemble-Level Entropy Tightness]
Let $\mathcal P_m$ denote a collection of prime multisets of fixed size $m$,
and define the empirical entropy distribution
\[
\nu_m := \frac{1}{|\mathcal P_m|}
\sum_{\mathbf p \in \mathcal P_m} \delta_{H(\mathbf p)}.
\]
Then the family $\{\nu_m\}_{m\ge 1}$ is tight.
Moreover, after centering by the Poisson null-model entropy,
$\nu_m$ converges weakly to a limiting probability measure as $m\to\infty$.
\end{conjecture}

\begin{remark}
The final conjecture concerns ensemble-level behavior and is independent
of the asymptotic stability of entropy for individual configurations.
It suggests the existence of a statistical limit theory for entropy
observables associated with structured point processes.
\end{remark}

\section{Broader Applicability and Significance}

Although motivated by analytic number theory, the construction applies to general point processes and discrete configurations exhibiting multiplicative or scale-free behavior.
Working at fixed logarithmic resolution is natural in settings where absolute scale varies widely or lacks intrinsic meaning.

Potential applications include spatial statistics, statistical physics, disordered systems, and the analysis of empirical point patterns.
Prime numbers serve as a well-studied benchmark rather than a necessary component of the framework.

\section*{Acknowledgments}

AI-based tools were used for exploratory assistance and editorial refinement.
All mathematical reasoning, results, and verification are due to the author.

\bigskip
\begingroup
\pagestyle{empty}
\printbibliography

@article{cramer_order_1936,
  author    = {Cram{\'e}r, Harald},
  title     = {On the order of magnitude of the difference between consecutive prime numbers},
  journal   = {Acta Arithmetica},
  volume    = {2},
  number    = {1},
  pages     = {23--46},
  year      = {1936},
  doi       = {10.4064/aa-2-1-23-46}
}

@article{gallagher_distribution_1976,
  author    = {Gallagher, P. X.},
  title     = {On the distribution of primes in short intervals},
  journal   = {Mathematika},
  volume    = {23},
  number    = {1},
  pages     = {4--9},
  year      = {1976},
  doi       = {10.1112/S0025579300005726}
}

@article{hardy_littlewood_1923,
  author    = {Hardy, G. H. and Littlewood, J. E.},
  title     = {Some problems of {Partitio Numerorum}; III. On the expression of a number as a sum of primes},
  journal   = {Acta Mathematica},
  volume    = {44},
  pages     = {1--70},
  year      = {1923},
  doi       = {10.1007/BF02412410}
}
\endgroup

\newpage
\appendix
\section{R Implementation}

\begin{remark}
The code below is intended for conceptual illustration.
For large-scale numerical experiments, optimized prime sieves and distance computations should be used.
\end{remark}

\begin{lstlisting}
# Prime generation
generate_primes <- function(n) {
  primes <- c(2)
  k <- 3
  while (length(primes) < n) {
    if (all(k %% primes != 0)) {
      primes <- c(primes, k)
    }
    k <- k + 2
  }
  primes
}

# Truncated distance multiset
prime_distances <- function(p, primes, R) {
  d <- abs(primes - p)
  d[d > 0 & d <= R]
}

# Logarithmic binning
log_bin <- function(distances, M) {
  if (length(distances) == 0) stop("No distances available")
  dmin <- min(distances)
  dmax <- max(distances)
  log_edges <- seq(log(dmin), log(dmax), length.out = M + 1)
  bins <- cut(distances, breaks = exp(log_edges),
              include.lowest = TRUE, labels = FALSE)
  counts <- tabulate(bins, nbins = M)
  probs <- counts / sum(counts)
  centers <- 0.5 * (log_edges[-1] + log_edges[-length(log_edges)])
  list(p = probs, centers = centers)
}

# Log-frequency spectrum
log_spectrum <- function(p, x) {
  M <- length(p)
  if (length(x) != M) stop("Length mismatch")
  denom <- x[M] - x[1]
  if (denom == 0) stop("Degenerate log-bin centers")
  sapply(seq_len(M), function(k)
    sum(p * exp(-2i * pi * (k - 1) * (x - x[1]) / denom))
  )
}

# Spectral entropy
spectral_entropy <- function(z, eps = 1e-12) {
  a <- Mod(z)
  s <- sum(a)
  if (s == 0) return(0)
  w <- a / s
  -sum(w * log(w + eps))
}

# Visualization
plot_spectrum <- function(z) {
  plot(Re(z), Im(z), asp = 1, pch = 19,
       col = heat.colors(length(z)),
       xlab = "Re", ylab = "Im",
       main = "Log-Frequency Spectrum")
  abline(h = 0, v = 0, lty = 2, col = "gray")
}

# Example run
set.seed(1)
primes <- generate_primes(10000)
p0 <- 101
R <- 5000
M <- 50
d <- prime_distances(p0, primes, R)
b <- log_bin(d, M)
spec <- log_spectrum(b$p, b$centers)
plot_spectrum(spec)
spectral_entropy(spec)
\end{lstlisting}

\bigskip
\begin{center}
{\small
Uppsala University, Uppsala, Sweden\\
\texttt{mohamed.gewily@uu.se}
}
\end{center}

\end{document}